\documentclass[aps,showkeys]{revtex4}
\usepackage{braket}
\usepackage{amsmath}
\usepackage{amssymb}
\usepackage{amsthm}
\usepackage{graphicx}
\usepackage{appendix}
\usepackage {xcolor}
\theoremstyle{definition}
\newtheorem{theorem}{Theorem}

\newtheorem{example}{Example}

\def\tr{\mathop{\rm Tr}}

\begin{document}
\title{A Note on Schmidt-number witnesses based on symmetric measurements}
\author{Xiao-Qian Mu}
\email{2230502152@cnu.edu.cn}
\affiliation{School of Mathematical Sciences, Capital Normal University, Beijing 100048, China}
\author{Hao-Fan Wang}
\email{2230502117@cnu.edu.cn}
\affiliation{School of Mathematical Sciences, Capital Normal University, Beijing 100048, China}
\author{Shao-Ming Fei}
\email{feishm@cnu.edu.cn}
\affiliation{School of Mathematical Sciences, Capital Normal University, Beijing 100048, China}
\begin{abstract}
The Schmidt number is an important kind of characterization of quantum entanglement. Quantum states with higher Schmidt numbers demonstrate significant advantages in various quantum information processing tasks. By deriving a class of $k$-positive linear maps based on symmetric measurements, we present new Schmidt-number witnesses of class $(k+1)$. By detailed example, we show that our Schmidt number witnesses identify better the Schmidt number of quantum states in high-dimensional systems. \textcolor{blue}{Furthermore, we note that the Fedorov ratio, which coincides with the Schmidt number for pure Gaussian states and provides a close approximation in non-Gaussian cases such as spontaneous parametric down-conversion (SPDC), serves as an experimentally accessible tool for validating the proposed $(k + 1)$-class Schmidt-number witnesses. }
\end{abstract}

\keywords{Schmidt-number witnesses, Symmetric measurements, Quantum entanglement}

\maketitle

\section{Introduction}

\indent	Entanglement is one of the most fundamental features in quantum mechanics \cite{R. Horodecki,M. B. Plenio}. As a key resource in quantum information processing, it enables quantum tasks such as quantum cryptography \cite{A. K. Ekert}, quantum teleportation \cite{C. H. Bennett} and superdense coding \cite{Bennett}, and plays vital roles in quantum computing and communication.

The Schmidt number is a fundamental quantity that reflects how many orthogonal quantum states contribute to the entanglement of a two-sided density matrix \cite{Terhal}. A higher Schmidt number implies a more complex and higher dimensional entanglement structure, thus providing greater potential for quantum information processing. Quantum states with high Schmidt number show significant advantages in such as quantum channel discrimination, communication, control, general-purpose computing and quantum key distribution from the view of key rate and fault tolerance \cite{Huber}.

Various methods have been developed to detect and characterize the Schmidt number of a given state. In 2000, Terhal et al.\cite{Terhal} proposed a renowned Schmidt number criterion by examining the fidelity between the quantum state and maximally entangled states. It is proved that the Schmidt number of a quantum state $\rho$ is greater than $k$ if and only if there exists a certain k-positive linear map $\Lambda$ such that $({\rm id}\otimes\Lambda)(\rho)$ is not positive semidefinite. In Refs.\cite{Hulpke,Johnston} the authors studied the Schmidt numbers based on the well-known positive partial transpose (PPT) entanglement criterion \cite{PPT1,PPT2} and the computable cross-norm (CCNR) or realignment criterion \cite{CCNR1,CCNR2}. The Schmidt number criteria based on Bloch decomposition \cite{Klockl}, covariance matrix \cite{Liu1,Liu2} and correlation trace norm \cite{Tavakoli,ZWang} have been also derived.

Another approach to detect the Schmidt numbers is to use the Schmidt number witnesses. In Ref.\cite{Sanpera} the authors introduced a canonical form of Schmidt number witnesses, providing a structured framework for the implementation and optimization of the Schmidt number witnesses. In Ref.\cite{Wyderka} the authors developed an iterative algorithm that finds Schmidt-number witnesses tailored to the measurements available in specific experimental setups. The author in Ref.\cite{Shi} provided a class of Schmidt number witnesses based on symmetric informationally complete measure (SIC POVM) and mutually unbiased bases (MUBs), which generalizes the results given in Ref.\cite{Chruściński}.
In fact, the mutually unbiased measurements (MUMs) \cite{Kalev1} and general symmetric informationally complete measure (GSIC POVM) \cite{Kalev2} are the natural extensions of MUBs and SIC POVM, respectively. Recently, a new kind of measurement called symmetric measurement or $(N,M)$-POVM, has been proposed in Ref.\cite{Siudzinska}, which can be thought of as a common generalization of GSIC POVM and MUMs. Recently, in Ref.\cite{HFWang} the authors derived a Schmidt number criterion via symmetric measurements.

Nevertheless, despite the advances above, reliable characterization of Schmidt numbers in high-dimensional systems remains a challenging and ongoing problem in the theory of quantum entanglement. In this paper, we present a class of $k$-positive maps by using the informationally complete $(N,M)$-POVM, as well as the corresponding Schmidt-number witness which include the entanglement witnesses given in Ref.\cite{Chruściński} as particular cases. \textcolor{blue}{Although our work is purely theoretical, it can be combined with an existing experimentally measurable entanglement quantifier, the Fedorov ratio $R_q$, which can be used together with our \((k+1)\)-class Schmidt-number witnesses. This combination provides a promising direction for future research, offering both strict lower bounds on the Schmidt number and direct experimental validation for high-dimensional entanglement studies\cite{Fedorov} \cite{G. Brida}.}

\section{The $k$-positive map and Schmidt-number witness based on $(N,M)$-POVM}
A pure bipartite state $\ket{\psi}$ in Hilbert space $\mathcal{H}_{A}\otimes\mathcal{H}_{B}$ has a Schmidt decomposition $ \ket{\psi}=\sum\limits_{i=1}^{r} \lambda_i \ket{u_{i}}\otimes\ket{v_{i}}$, where $\lambda_i>0$, $\sum\limits_{i=1}^{r}\lambda_i^2=1$, $\{\ket{u_i}\}$ and $\{\ket{v_i}\}$ are the orthonormal bases in $\mathcal{H}_{A}$ and $\mathcal{H}_{B}$, respectively, and $r\equiv{\rm SR}(\ket{\psi})$ is the Schmidt rank of $\ket{\psi}$. The Schmidt number of a general bipartite state $\rho$ in $\mathcal{H}_{A}\otimes\mathcal{H}_{B}$ is defined as \cite{Terhal},
\begin{equation*}
	{\rm SN}(\rho) = \min\limits_{\{p_{i},\ket{\psi_{i}}\}}\max\limits_{i}{\rm SR}(\ket{\psi_{i}}),
\end{equation*}
where the minimization takes over all pure state decompositions of $\rho=\sum_i p_i \ket{\psi_i}\bra{\psi_i}$.

Denote $S_{k}=\{\rho\in\mathcal{D}(\mathcal{H}_{A}\otimes\mathcal{H}_{B})|{\rm SN}(\rho)\leq k\}$, where $\mathcal{D}(\mathcal{H}_{A}\otimes\mathcal{H}_{B})$ is the set of all density operators on $\mathcal{H}_{A}\otimes\mathcal{H}_{B}$. It is easy to see that $S_1\subseteq S_{2}\subset\cdots\subseteq S_{k}\subseteq\cdots$. In particular, $S_{1}$ is the set of separable states, since a bipartite state is separable if and only if its Schmidt number is one.

A Hermitian operator $W_{k}$ acting on $\mathcal{H}_{A}\otimes\mathcal{H}_{B}$ is called a Schmidt-number witness of class $(k+1)$ if $\tr(W_k\rho)\geq 0$ for all $\rho\in S_{k}$ and $\tr(W_k\sigma)<0$ for at least one $\sigma\in\mathcal{D}(\mathcal{H}_{A}\otimes\mathcal{H}_{B})$.
A linear Hermitticity-preserving map $\Lambda$ is said to be $k$-positive if and only if
$(I\otimes\Lambda)(\ket{\psi_k}\bra{\psi_k})\geq 0$, where $\ket{\psi_k}=(U\otimes V)\sqrt{\frac{1}{k}}\sum\limits_{i=0}^{k-1}\ket{ii}$ is any maximally entangled state of Schmidt number $k$, with $U$ and $V$ being arbitrary unitary operators of $\mathcal{H}_A$ and $\mathcal{H}_B$, respectively. If $\Lambda$ is $k$-positive, then $\Lambda^{\dagger}$ defined by $\tr(\rho_1^{\dagger}\cdot\Lambda(\rho_2))=\tr(\Lambda^{\dagger}(\rho_1^{\dagger})\cdot \rho_2)$ for all density operators $\rho_1$ and $\rho_2$, is also $k$-positive, where $\dagger$ stands for complex conjugation and transpose.

An $(N,M)$-POVM \cite{Siudzinska} is a collection of $N$ $d$-dimensional POVMs $\{E_{\alpha,k}|k=1,2\cdots,M\}$ ($\alpha=1,2,\cdots,N$) satisfying the following conditions,
\begin{flalign*}
	\tr(E_{\alpha,k}) &= \dfrac{d}{M}, \\
	\tr(E_{\alpha,k}^{2}) &= x,\\
	\tr(E_{\alpha,k}E_{\alpha,l}) &= \dfrac{d-Mx}{M(M-1)},~~ l\neq k\\
	\tr(E_{\alpha,k}E_{\beta,l}) &= \dfrac{d}{M^{2}},~~ \beta\neq\alpha
\end{flalign*}
where $\dfrac{d}{M^{2}}<x\leq\min\left\{\dfrac{d^{2}}{M^{2}},\dfrac{d}{M}\right\}$. Especially, an $(N,M)$-POVM with $N(M-1)=d^{2}-1$ is called an informationally complete $(N,M)$-POVM. For any finite dimension $d$ ($d>2$), there exist at least four distinct classes of informationally complete $(N,M)$-POVM: (1) $N=1$ and $M=d^{2}$ (GSIC POVM), (2) $N=d+1$ and $M=d$ (MUMs), (3) $N=d^{2}-1$ and $M=2$, (4) $N=d-1$ and $M=d+2$.\par

From orthonormal Hermitian operator basis $\{G_{0}=I_{d}/\sqrt{d},\,G_{\alpha,k}|\alpha=1,\cdots,N;\,k=1,\cdots,M-1\}$ with ${\rm tr}(G_{\alpha,k})=0$, an informationally complete $(N,M)$-POVM is given by
\begin{equation*}
	E_{\alpha,k}=\dfrac{1}{M}I_{d}+tH_{\alpha,k},
\end{equation*}
where
\begin{equation*}
	H_{\alpha,k}=\begin{cases}
		G_{\alpha}-\sqrt{M}(\sqrt{M}+1)G_{\alpha,k},~ k=1,\cdots,M-1\\
		(\sqrt{M}+1)G_{\alpha}, k=M
	\end{cases}
\end{equation*}
with $G_{\alpha}=\sum\limits_{k=1}^{M-1}G_{\alpha,k}$, the parameter $t$ should be chosen such that $E_{\alpha,k}\geq 0$, namely,
\begin{equation*}
	-\dfrac{1}{M}\dfrac{1}{\lambda_{\max}}\leq t \leq \dfrac{1}{M}\dfrac{1}{|\lambda_{\min}|},
\end{equation*}
with $\lambda_{\max}$ and $\lambda_{\min}$ being the minimal and maximal eigenvalue among all the eigenvalues of $H_{\alpha,k}$, respectively.
The parameters $t$ is related to $x$ via $x=\dfrac{d}{M^{2}}+t^{2}(M-1)(\sqrt{M}+1)^{2}$.

Let $\{E_{\alpha,k}|\alpha=1,\cdots,N;\,k=1,\cdots,M\}$ be an informationally complete $(N, M)$-POVM on $d$ dimensional Hilbert space $\mathcal{H}$ with free parameter $x$, and $\mathcal{O}^{(\alpha)}=(O_{kl}^{(\alpha)})$ ($\alpha=1,2\cdots,N$) be a family of $M\times M$ orthogonal rotation matrices which preserve the vector $\mathbf{n}_{*}=\frac{(1,1,\cdots,1)^{\mathsf{T}}}{\sqrt{d}}$, where $\mathsf{T}$ denotes transpose. We have the following class of $k$-positive maps $\Lambda$.

\begin{theorem}
The following linear map
\begin{equation}\label{1}
\Lambda(X)=  \frac{\mathbb{I}_d}{d}{\rm Tr}(X)-h_{x}\sum\limits_{\alpha=1}^{N}\sum\limits_{g,l=1}^{M}\mathcal{O}_{gl}^{(\alpha)}{\rm Tr}[(X-\frac{\mathbb{I}_d}{d}{\rm Tr}(X))E_{\alpha,l}]E_{\alpha,g}
\end{equation}
is $k$-positive for any density operator $X$, where $h_{x}=\frac{1}{kd-1}\sqrt{\frac{M(M-1)}{x(M^2x-d)}}$ and $\mathbb{I}_d$ is the $d\times d$ identity operator.
\end{theorem}

\begin{proof}
It has been shown that ${\rm SN}(\rho)\geq k+1$ if and only if there exists a $k$-positive linear map $\Lambda_k$ such that $(I\otimes\Lambda_k)(\rho)\ngeq0$ Ref.\cite{Terhal}. Therefore, we complete the proof by proving that $(I\otimes \Lambda)(\ket{\psi_k}\bra{\psi_k})\geq 0$ for all maximally entangled state $\ket{\psi_k}$ of Schmidt number $k$.
Since ${\rm tr}(\rho^2)\le {1}/{d-1}$ for any state $\rho$ \cite{Bengtsson},
to show the $k$-positivity of $\Lambda$ we only need to show that \cite{Chruściński}
$\tr[(I\otimes \Lambda)(\ket{\psi_k}\bra{\psi_k})]^2\le {1}/{d-1}$ for all maximally entangled state $\ket{\psi_k}$ of Schmidt number $k$.

Noting that $\ket{\psi_k}=(U\otimes V)\sqrt{\frac{1}{k}}\sum\limits_{i=0}^{k-1}\ket{ii}$ with $U$ and $V$ being arbitrary unitary operators of $\mathcal{H}_A$ and $\mathcal{H}_B$, respectively, we have
    \begin{flalign}\label{3}
    &\tr[(I\otimes \Lambda)(\ket{\psi_k}\bra{\psi_k})]^2\nonumber\\
    =&\tr\frac{1}{k^2}[(I\otimes\Lambda)(U\otimes V)(\sum_{i=0}^{k-1}\sum_{j=0}^{k-1}\ket{ii}\bra{jj})(U^{\dagger}\otimes V^{\dagger})]^2\nonumber\\
    =&\tr\frac{1}{k^2}[\sum_{i,j=0}^{k-1}U\ket{i}\bra{j}U^{\dagger}\otimes\Lambda(V\ket{i}\bra{j}V^{\dagger})]^2\nonumber\\
    =&\tr\frac{1}{k^2}[\sum_{i,j,m=0}^{k-1}U\ket{i}\bra{m}U^{\dagger}\otimes \Lambda(V\ket{i}\bra{j}V^{\dagger})\Lambda(V\ket{j}\bra{m}V^{\dagger})]\nonumber\\
    =&\frac{1}{k^2}\tr\left(\sum_{i,j=0}^{k-1}\Lambda(V\ket{i}\bra{j}V^{\dagger})\Lambda(V\ket{j}\bra{i}V^{\dagger})\right).
    \end{flalign}
If $i=j$, we have
    \begin{equation}\label{c1}
    \Lambda(V\ket{i}\bra{i}V^{\dagger})=\frac{\mathbb{I}_d}{d}-h_{x}
    \sum_{\alpha=1}^{N}\sum_{g,l=1}^{M}\mathcal{O}_{gl}^{(\alpha)}
    \tr[(V\ket{i}\bra{i}V^{\dagger}-\frac{\mathbb{I}_d}{d})E_{\alpha,l}]E_{\alpha,g},
    \end{equation}
and
\begin{equation}\label{c2}
    \Lambda(V\ket{i}\bra{i}V^{\dagger})=-h_{x}\sum_{\alpha=1}^{N}
    \sum_{g,l=1}^{M}\mathcal{O}_{gl}^{(\alpha)}\tr[(V\ket{i}\bra{j}V^{\dagger})
    E_{\alpha,l}]E_{\alpha,g}
    \end{equation}
if $i\ne j$.

Substituting (\ref{c1}) and (\ref{c2}) into (\ref{3}), we obtain
    \begin{flalign}\label{6}
    (\ref{3})=&\frac{1}{kd}+\frac{h_{x}^2}{k^2}\sum\limits_{i=0}^{k-1}\sum\limits_{\alpha,\beta=1}^{N}\sum\limits_{g,l,s,t=1}^{M}\mathcal{O}_{gl}^{(\alpha)}\mathcal{O}_{ts}^{(\beta)}\tr[(V\ket{i}\bra{i}V^{\dagger}-\frac{\mathbb{I}_d}{d})E_{\alpha,l}]\tr[(V\ket{i}\bra{i}V^{\dagger}-\frac{\mathbb{I}_d}{d})E_{\beta,s}]\tr(E_{\alpha,g}E_{\beta,t})\nonumber\\
    -&\frac{2h_{x}}{k^2}\sum\limits_{i=0}^{k-1}\tr\sum\limits_{\alpha=1}^{N}\sum\limits_{g,l=1}^{M}\mathcal{O}_{gl}^{(\alpha)}\tr[(V\ket{i}\bra{i}V^{\dagger}-\frac{\mathbb{I}_d}{d})E_{\alpha,l}]\tr(E_{\alpha,g})\nonumber\\
    +&\frac{h_{x}^2}{k^2}\sum\limits_{i\ne j}\sum\limits_{\alpha,\beta=1}^{N}\sum\limits_{g,l,s,t=1}^{M}\mathcal{O}_{gl}^{(\alpha)}\mathcal{O}_{st}^{(\beta)}\tr[V\ket{i}\bra{j}V^{\dagger}E_{\alpha,l}]\tr[V\ket{j}\bra{i}V^{\dagger}E_{\beta,s}]\tr(E_{\alpha,g}E_{\beta,t})\nonumber\\
    =&\frac{1}{kd}+\frac{h_{x}^2}{k^2}\sum\limits_{i=0}^{k-1}\sum\limits_{\alpha=1}^{N}\sum\limits_{g,l,s,t=1}^{M}\mathcal{O}_{gl}^{(\alpha)}\mathcal{O}_{ts}^{(\alpha)}\tr[(V\ket{i}\bra{i}V^{\dagger}-\frac{\mathbb{I}_d}{d})E_{\alpha,l}]\tr[(V\ket{i}\bra{i}V^{\dagger}-\frac{\mathbb{I}_d}{d})E_{\alpha,s}]\tr(E_{\alpha,g}E_{\alpha,t})  \nonumber\\
    +&\frac{h_{x}^2}{k^2}\sum\limits_{i=0}^{k-1}\sum\limits_{\alpha \ne\beta}\sum\limits_{g,l,s,t=1}^{M}\mathcal{O}_{gl}^{(\alpha)}\mathcal{O}_{ts}^{(\beta)}\tr[(V\ket{i}\bra{i}V^{\dagger}-\frac{\mathbb{I}_d}{d})E_{\alpha,l}]\tr[(V\ket{i}\bra{i}V^{\dagger}-\frac{\mathbb{I}_d}{d})E_{\beta,s}]\tr(E_{\alpha,g}E_{\beta,t})      \nonumber\\
    +&\frac{h_{x}^2}{k^2}\sum\limits_{i\ne j}\sum\limits_{\alpha=1}^{N}\sum\limits_{g,l,s,t=1}^{M}\mathcal{O}_{gl}^{(\alpha)}\mathcal{O}_{st}^{(\alpha)}\tr[V\ket{i}\bra{j}V^{\dagger}E_{\alpha,l}]\tr[V\ket{j}\bra{i}V^{\dagger}E_{\alpha,s}]\tr(E_{\alpha,g}E_{\alpha,t})     \nonumber\\
    +&\frac{h_{x}^2}{k^2}\sum\limits_{i\ne j}\sum\limits_{\alpha\ne\beta}\sum\limits_{g,l,s,t=1}^{M}\mathcal{O}_{gl}^{(\alpha)}\mathcal{O}_{st}^{(\beta)}\tr[V\ket{i}\bra{j}V^{\dagger}E_{\alpha,l}]\tr[V\ket{j}\bra{i}V^{\dagger}E_{\beta,s}]\tr(E_{\alpha,g}E_{\beta,t})       \nonumber\\
    =&\frac{1}{kd}+\frac{h_{x}^2x}{k^2}\sum_{i=0}^{k-1}\sum\limits_{\alpha=1}^{N}\sum\limits_{l=1}^{M}|\tr[(V\ket{i}\bra{i}V^{\dagger}-\frac{\mathbb{I}_d}{d})E_{\alpha,l}]|^2+\frac{h_{x}^2x}{k^2}\sum_{i\ne j}\sum\limits_{\alpha=1}^{N}\sum\limits_{l=1}^{M}\tr(E_{\alpha,l} V\ket{i}\bra{j}V^{\dagger})\tr(E_{\alpha,l} V\ket{j}\bra{i}V^{\dagger}).\nonumber
        \end{flalign}

By using the following relation (Lemma 1 in Ref.\cite{HFWang}),
\begin{equation*}
\sum_{\alpha=1}^{N}\sum_{k=1}^{M}|{\rm tr}(E_{\alpha,k}\sigma)|^{2}=\dfrac{d(M^{2}x-d){\rm tr}(\sigma\sigma^{\dagger})+(d^{3}-M^{2}x)|{\rm tr}(\sigma)|^{2}}{dM(M-1)}
\end{equation*}
for any density operator $\sigma$, we have
    \begin{flalign}
    &\tr[(I\otimes \Lambda)(\ket{\psi_k}\bra{\psi_k})]^2\nonumber\\
    =& \frac{1}{kd}+\frac{h_{x}^2x(M^2x-d)(d-1)}
    {kdM(M-1)}+\frac{h_{x}^2x(k-1)}{k}\frac{M^2x-d}{M(M-1)}\nonumber\\
    =&\frac{M(M-1)+h_{x}^2x(M^2x-d)(d-1)+h_{x}^2x(k-1)d(M^2x-d)}{kdM(M-1)}\nonumber\\
    =&\frac{1}{kd-1}\leq\frac{1}{d-1},\nonumber
    \end{flalign}
which proves that $\Lambda(X)$ is a $k$-positive linear map.
\end{proof}

Based on the Theorem above, we have the following class of witnesses:

\begin{theorem}
The operators $W_k$ below are Schmidt-number witnesses of class $(k+1)$,
\begin{equation}
W_k=\frac{M+Ndh_{x}}{dM}\mathbb{I}_d\otimes \mathbb{I}_d-h_{x}\sum\limits_{\alpha=1}^{N}\sum\limits_{g,l=1}^{M}
\mathcal{O}_{gl}^{(\alpha)}\overline{E_{\alpha,l}}\otimes E_{\alpha,g},\label{w1}
\end{equation}
namely, $\tr(W_k\rho)\geq 0$ for all $\rho\in S_{k}$.
\end{theorem}

In particular, when $N-1=M=d$, $x=1$ and $k=1$, (\ref{w1}) reduces to
    \begin{equation*}
    W_k=\frac{2}{d-1}\mathbb{I}_d\otimes \mathbb{I}_d-\frac{1}{d-1}\sum_{\alpha=1}^{N}\sum_{g,l=1}^{M}
    \mathcal{O}_{gl}^{(\alpha)}\overline{E_{\alpha,l}}\otimes E_{\alpha,g},
    \end{equation*}
which gives rise to the entanglement witnesses given in Eq.(15) of Ref.\cite{Chruściński}.

\textcolor{blue}{\textbf{Remark (Connection to the Fedorov ratio).} The witness $W_k$ provides a strict lower bound for the Schmidt number $ SN(\rho)\geq k + 1\ $. If the experimental state is close to a pure Gaussian state, such as those generated in SPDC, the Fedorov ratio $R$ is approximately equal to the Schmidt number $K$ (i.e., $R\approx K$ ). When the measured value of $ R \geq k + 1 $, it can be used to validate the negative expectation value of $W_k$, forming a \textit{dual validation} approach through both theoretical witnesses and experimentally measurable quantities. In more general scenarios, $R$ serves as an upper bound or approximation for $K$, providing an alternative check for the Schmidt number in experimental contexts. \cite{Fedorov} \cite{G. Brida}}.

\begin{example}
Consider the following isotropic state \cite{isostate},
\begin{equation*}
    \rho_v = v\ket{\psi_d}\bra{\psi_d}+\frac{1-v}{d^2}I_{d^2}
\end{equation*}
where $\ket{\psi_d}=\frac{1}{\sqrt{d}}\sum_{i}\ket{ii}$ is the maximally entangled state and $v \in{(0,1]}$.

Denote $\alpha=\frac{M+Ndh_{x}}{dM}$ and $\Psi_d=\ket{\psi_d}\bra{\psi_d}$. From our Schmidt witness (\ref{w1}) we obtain
  \begin{flalign}\label{e6}
       \tr (W_k {\rho_v})=\tr&[(\alpha I_{d^2}-h_x\sum_{\alpha=1}^N\sum_{g,l=1}^M {O}_{gl}^{(\alpha)}\overline{E_{\alpha,l}}\otimes E_{\alpha,g})(v\Psi_d+\frac{1-v}{d^2}I_{d^2})]\nonumber\\
      =&\alpha v+\alpha(1-v)-\frac{h_x(1-v)N}{M}-vh_x\sum_{\alpha=1}^N\sum_{g,l=1}^M \tr[({O}_{gl}^{(\alpha)}\overline{E_{\alpha,l}}\otimes E_{\alpha,g})\Psi_d]\nonumber\\
       =&\alpha v+\alpha(1-v)-\frac{h_x(1-v)N}{M}-vh_x\sum_{\alpha=1}^N\sum_{g,l=1}^M\tr[{O}_{gl}^{(\alpha)}(I_{d}\otimes E_{\alpha,l}^{\dagger}E_{\alpha,g})\Psi_d]                  \nonumber\\
       =&\alpha v+\alpha(1-v)-\frac{h_x(1-v)N}{M}-vh_x\frac{1}{d}\sum_{\alpha=1}^N\sum_{g,l=1}^M\tr({E_{\alpha,g}E_{\alpha,l}})         \nonumber\\
        =&\alpha v+\alpha(1-v)-\frac{h_x(1-v)N}{M}-vh_x\sum_{\alpha=1}^N\sum_{g=1}^M\frac{{O}_{gg}^{(\alpha)}}{d}\tr({E_{\alpha,g}^2})-vh_x\frac{d-Mx}{M(M-1)}\frac{1}{d}\sum_{\alpha=1}^N\sum_{g\neq l}{O}_{gl}^{(\alpha)}        \nonumber\\
         =&\alpha v+\alpha(1-v)-\frac{h_x(1-v)N}{M}-vh_x[(x-\frac{d-Mx}{M(M-1)})\sum_{\alpha=1}^N\sum_{g=1}^M\frac{{O}_{gg}^{(\alpha)}}{d}+\frac{d-Mx}{M(M-1)}\frac{1}{d}\sum_{\alpha=1}^N\sum_{g,l=1}^M{O}_{gl}^{(\alpha)}]            \nonumber\\
         =&\alpha v+\alpha(1-v)-\frac{h_x(1-v)N}{M}-vh_x[(x-\frac{d-Mx}{M(M-1)})
         \sum_{\alpha=1}^N\sum_{g=1}^M\frac{{O}_{gg}^{(\alpha)}}{d}
         +\frac{d-Mx}{M-1}\frac{1}{d}N],
  \end{flalign}
where in the third equality, we have used the following fact, $(A\otimes I)\ket{\psi_d}=(I\otimes A^{\mathsf{T}})\ket{\psi_d}$ for any linear operator A, with $A=\overline{E_{\alpha,l}}$, $I=I_{d}$ and $\ket{\psi_d}=\frac{1}{\sqrt{d}}\sum_{i}\ket{ii}$.

In particular, if we take ${O}_{gg}^{(\alpha)}=1$ ($g=1,\cdots,M$; $\alpha=1,\cdots,N$), we obtain
\begin{equation*}
\tr (W_k {\rho_v})=\alpha-\frac{Nh_x(1-v)}{M}-\frac{xvh_xNM}{d}.
\end{equation*}
Therefore, when
 \begin{equation*}
 v>\frac{M(kd-1){\sqrt{x(M^2x-d)}}}{(M^2Nx-Nd)\sqrt{M(M-1)}},
 \end{equation*}
the Schmidt number of $\rho_v$ is greater than $k$.

Let us take $d=3$, $k=2=M=2$ and $N=8$. Then $x \in(0.75,1.50]$.
Based on Schmidt witness given in Ref.\cite{Shi}, by taking $L=4$ and correcting a calculation error in Ref.\cite{Shi}, we get $SN(\rho_v)>2$ when $v>\frac{\sqrt{(dk-1)(Lk-L+d-1)}}{dL-L}=0.685$.
From our theorem we have that $SN(\rho_v)>2$ when $v>{\frac{5}{\sqrt{2}}}\frac{\sqrt{x(4x-3)}}{16x-12}\equiv v(x)$.
From Fig.1 we see that $SN(\rho_v)>2$ when $v>0.625$ ($x=1.5$).
In fact, when $x>1.286$, our witness already gives rise to the result that
$SN(\rho_v)>2$ when $v>0.685$ obtained by \cite{Shi}.
Hence, our Schmidt witness is more efficient in detecting $SN(\rho_v)$.
\begin{figure}[h]
\centering
\includegraphics[width=0.7\textwidth]{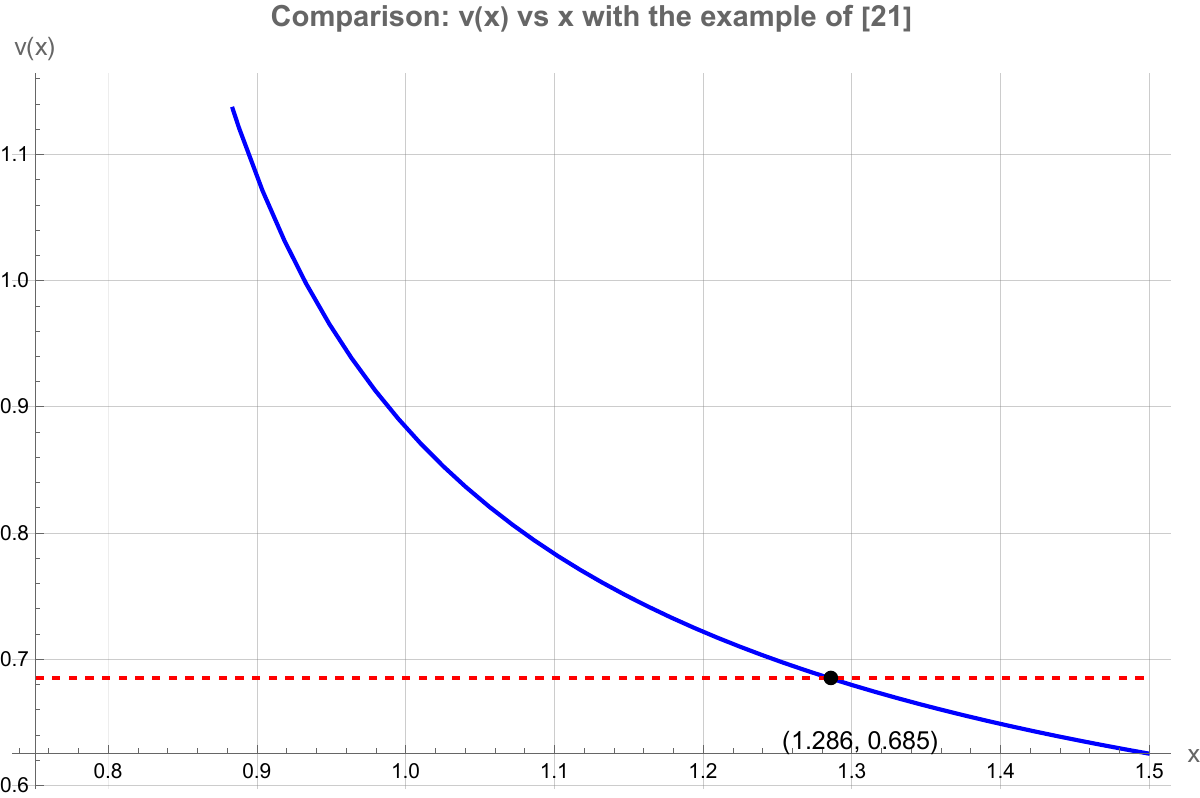}
\caption{comparison}
The red line represents the result $v=0.685$ presented in example 1 in \cite{Shi}. The blue line represents the function $v(x)$ of $x$, with the smallest value $0.625$ of $v(x)$ attained at $x=1.5$.
\label{fig:myfigure}
\end{figure}
\end{example}

\section{Conclusions and Discussions}
We have presented a Schmidt number witness of class $(k+1)$ by constructing $k$-positive linear maps based on $(N,M)$-POVM. By detailed example, we have shown that our Schmidt number witness detects better Schmidt number of quantum states in high-dimensional systems. Our approach may inspire the constructions of new Schmidt witnesses and supply a plausible way to verify experimentally the Schmidt number of given quantum states.

\textcolor{blue}{Based on current frequency-domain SPDC experiments, we suggest that for a given source, the following two steps should be performed simultaneously:  
(i) measure the $(N,M)$-POVM statistics and compute $\text{Tr}(W_k \rho)$;  
(ii) measure the Fedorov ratio $R$, which is related to the Schmidt number $K$. For pure or near-pure states, $R$ and $K$ are closely related. When (i) provides a lower bound $SN(\rho) \geq k+1$ and (ii) shows $R \geq k+1$, consistent experimental and theoretical evidence can be established. }

\bigskip
\noindent{\bf Acknowlegements}
This work is supported by the National Natural Science Foundation of China (NSFC) under Grant No. 12171044, the specific research fund of the Innovation Platform for Academicians of Hainan Province.

\end{document}